\documentclass[11pt]{article}

\usepackage[margin=1in]{geometry}

\usepackage{amsmath}
\usepackage{amsthm}
\usepackage{lipsum}
\usepackage{amsfonts}
\usepackage{graphicx}
\usepackage{epstopdf}
\usepackage{algorithmic}
\usepackage{xspace}
\usepackage{hyperref}
\usepackage{xcolor}
\usepackage{authblk}
\usepackage{cleveref}
\usepackage{subcaption}

\newtheorem{theorem}{Theorem}
\newtheorem{lemma}{Lemma}
\newtheorem{corollary}{Corollary}
\theoremstyle{definition}
\newtheorem{definition}{Definition}
\newtheorem{problem}{Problem}
\newtheorem{fact}{Fact}
\newtheorem{openproblem}{Open Problem}

\newcommand\smlg{\textsf{SMLG}\xspace}
\newcommand\pattern{\textsf{PATTERN}\xspace}

\newcommand\ov{\textsf{OV}\xspace}
\newcommand\ovh{\textsf{OVH}\xspace}

\newcommand\seth{\textsf{SETH}\xspace}
\newcommand\sic{\textsf{SIC}\xspace}

\newcommand\B{\ensuremath{\mathtt{b}}\xspace}
\newcommand\E{\ensuremath{\mathtt{e}}\xspace}

\newcommand\zero{\ensuremath{\mathtt{0}}\xspace}
\newcommand\one{\ensuremath{\mathtt{1}}\xspace}

\renewcommand\epsilon{\varepsilon}
\newcommand\atilde{\tilde{\text{a}}}
\newcommand\btilde{\tilde{\text{b}}}

\title{Graphs cannot be indexed in polynomial time for sub-quadratic time string matching, unless SETH fails}

\author{Massimo Equi}
\author{Veli M\"{a}kinen}
\author{Alexandru~I.~Tomescu}
\affil{Department of Computer Science, University of Helsinki\\\texttt{\{massimo.equi,veli.makinen,alexandru.tomescu\}@helsinki.fi}}

\begin{document}

\maketitle

\begin{abstract}
We consider the following string matching problem on a node-labeled graph $G=(V,E)$: given a pattern string $P$, decide whether there exists a path in $G$ whose concatenation of node labels equals $P$. This is a basic primitive in various problems in bioinformatics, graph databases, or networks. The hardness results of Backurs and Indyk (FOCS 2016) imply that this problem cannot be solved in better than $O(|E||P|)$ time, under the Orthogonal Vectors Hypothesis (\ovh), and this holds even under various restrictions on the graph (Equi et al.,~ICALP 2019). 

In this paper we consider its \emph{offline} version, namely the one in which we are allowed to index the graph in order to support time-efficient string matching queries. In fact, this version is the one relevant in practical application such as the ones mentioned above. While the online version has been believed to be hard even before the above-mentioned hardness results, it was tantalizing in the string matching community to believe that the offline version can allow for sub-quadratic time queries, e.g. at the cost of a high-degree polynomial-time indexing.

We disprove this belief, showing that, under \ovh, no polynomial-time indexing scheme of the graph can support querying $P$ in time $O(|E|^\delta|P|^\beta)$, with either $\delta < 1$ or $\beta < 1$. We prove this \emph{tight} bound employing a known self-reducibility technique, e.g.~from the field of dynamic algorithms, which translates conditional lower bounds for an online problem to its offline version. 

As a side-contribution, we formalize this technique with the notion of \emph{linear independent-components reduction}, allowing for a simple proof of our result. As another illustration that hardness of indexing follows as a corollary of a linear independent-components reduction, we also translate the quadratic conditional lower bound of Backurs and Indyk (STOC 2015) for the problem of matching a query string inside a text, under edit distance. We obtain an analogous \emph{tight} quadratic lower bound for its offline version, improving the recent result of Cohen{-}Addad, Feuilloley and Starikovskaya (SODA 2019), but with a slightly different boundary condition.
\end{abstract}

\newpage

\section{Introduction}

\subsection{Background}

The \emph{String Matching in Labeled Graphs (\smlg)} problem is defined as follows.
\begin{problem}[\smlg]
\item{\textsc{Input}:} A directed graph $G = (V,E,\ell)$, where each node $v \in V$ is labeled by a character $\ell(v)$, and a pattern string $P$.
\item{\textsc{Output}:} \emph{True} if and only if there is path $(v_1,v_2,\ldots,v_{|P|})$ in $G$ such that $P[i]=\ell(v_i)$ holds for all $1\leq i \leq |P|$.
\end{problem}

This is a natural generalization of the problem of matching a string inside a text, and it is a primitive in various problems in computational biology, graph databases, and graph mining. In genome research, the very first step of many standard analysis pipelines of high-throughput sequencing data is nowadays to align sequenced fragments of DNA on a labeled graph (a so-called \emph{pan-genome}) that encodes all genomes of a population~\cite{Sch09,Maretal18,hisat2,vg}. In graph databases, query languages provide the user with the ability to select paths based on the labels of their nodes or edges~\cite{AnglesGutierrez2008,FrancisGGLLMPRS18,Rodriguez15,sparqlquery}. In graph mining, this is a basic ingredient related to computing graph kernels~\cite{HidoK09} or node similarity~\cite{ConteFGMSU18}. 

The \smlg problem can be solved in time $O(|V| + |E||P|)$~\cite{amir1997pattern} in the comparison model. On acyclic graphs, bitparallelism can be used for improving the time to $O(|V| + |E|\lceil |P|/w\rceil)$~\cite{RMM19} in the RAM model with word size $w=\Theta(\log |E|)$. It remained an open question whether a truly sub-quadratic time algorithm for it exists. However, the recent conditional lower bounds by Backurs and Indyk \cite{BI16} for regular expression matching imply that the \smlg problem cannot be solved in sub-quadratic time, unless the so-called \emph{Orthogonal Vectors Hypothesis (\ovh)} is false. This result was strengthened by Equi et al.~\cite{Equi18} by showing that the problem remains quadratic under \ovh also on directed acyclic graphs (DAGs), that are even \emph{deterministic}, in the sense that for every node, the labels of its out-neighbors are all distinct.

As mentioned above, in real-world applications one usually considers the \emph{offline} version of the \smlg problem. Namely, we are allowed to index the labeled graph so that we can query for pattern strings in possibly sub-quadratic time. In the case when the graph is just a labeled (directed) path, then the problem asks about indexing a text string, which is a fundamental problem in string matching. There exists a variety of indexes constructable in \emph{linear time} supporting \emph{linear-time} queries \cite{jewels}. The same holds also when the graph is a tree~\cite{FMMN09}. A trivial indexing scheme for arbitrary graphs is to enumerate all the possibly exponentially many paths of the graph and index those as strings. So a natural question is whether we can at least index the graph in polynomial time to support sub-quadratic time queries. Note that the conditional lower bound for the online problem naturally refutes the possibility of an index constructable in sub-quadratic time to support sub-quadratic time queries. Even before the \ovh -based reductions, another weak lower bound was known to hold conditioned on the hardness of indexing for set intersection queries~\cite{Bil13} (see also \Cref{tab:bounds}). We discuss this connection to the \emph{Set Intersection Conjecture (\sic)} \cite{PR14,GLP19} in Appendix~\ref{app:sic}.

The connections to \sic and to \ovh constrain the possible construction and query time tradeoffs for \smlg, but they are yet not strong enough to prove the impossibility of building an index in polynomial time such that queries could be sub-quadratic, or even take time say $O(|E|^{1/2}|P|^2)$. This would be a significant result. In fact, given the wide applicability of this problem, there have been many attempts to obtain such indexing schemes. Sir\'en, V\"alim\"aki, and M\"akinen \cite{SVM14} proposed an extension of the \emph{Burrows-Wheeler transform} \cite{BW94} for prefix-sorted graphs. Standard indexing techniques \cite{GV05,FM05,NM07} can be applied on such generalized Burrows-Wheeler transform to support linear time pattern search. The bottleneck of the approach is the prefix-sorting step, which requires finding shortest prefixes for all paths such that they distinguish the nodes from each other. The size of the transform is still exponential in the worst case. However, unlike the trivial indexing scheme, it is linear in the best case, and also linear in the average case under a realistic model for genomics applications \cite{SVM14}.
There have been some advances in making the approach more practical \cite{Sir17,hisat2,vg}, but the exponential bottleneck has remained. Since in real-world scenarios approximate search is required on the graph, there have also been advances in expanding sparse dynamic programming and chaining algorithms \cite{Maketal19}, as well as the seed-and-extend strategy \cite{Eggertsson2017,RM19} to this setting. 

The concept of prefix-sorted graphs was later formalized into a more general concept of \emph{Wheeler graphs} \cite{GMS17}: Conceptually these are a class of graphs that admit a generalization of the Burrows-Wheeler transform, and thus an index of size linear in the size of the graph, supporting string search in linear time in the size of the query pattern.  Gibney and Thankachan showed that Wheeler graph recognition problem is NP-complete \cite{GT19}. Alanko et al.~\cite{ADPP20} give polynomial time solutions on some special cases and improve the prefix-sorting algorithm to work in near-optimal time in the size of the output. They also give an example where such output can be of exponential size even for \emph{acyclic deterministic finite automata} (acyclic DFA).
One could conjecture that conversion of a graph into an equivalent Wheeler graph is equally hard as indexing a graph for linear time string search, but as far as we know, such equivalence result has not yet been established. Therefore the hardness of indexing graphs is largely still open.

In this paper we refute the existence of such a polynomial indexing scheme for graphs, under \ovh. This contributes to a growing number of conditional lower bounds for offline string problem, such as the one for indexed \emph{jumbled pattern matching}~\cite{ACLL14}, conditioned on 3SUM-hardness, and the one for indexed \emph{approximate pattern matching under $\kappa$ differences}~\cite{CFS19}, conditioned on \seth.

Our result holds even for deterministic DAGs with labels from binary alphabet. By introducing a super-source connected to all source nodes and moving labels to incoming edges, such graphs can be interpreted as \emph{acyclic non-deterministic finite automata} (acyclic NFA) whose only non-deterministic state is the start state. It follows that determinisation of such simple NFAs cannot be done in polynomial time unless \ovh is false (and thus also unless \seth is false). This corollary complements the current picture of the exponential gap between NFAs and DFAs.

Table~\ref{tab:bounds} and Figure~\ref{fig:AlphaBetaDelta} summarize the complexity landscape around offline \smlg.

\begin{table}[t]
    \centering
    \begin{tabular}{c|c|c|c}
        \textbf{Graph} & \textbf{Indexing time} & \textbf{Query time} & \textbf{Reference, Year}\\\hline
        path & $O(|E|)$ & $O(|P|)$ & classical~\cite{jewels} \\\hline
        tree & $O(|E|)$ & $O(|P|)$ & \cite{FMMN09}, 2009\\\hline
        Wheeler graph & $O(|E|)$ & $O(|P|)$ & \cite{SVM14,GMS17}, 2014 \\\hline
        DAG & $O(|E|^{\alpha})$, $\alpha < 2$ & 
        \parbox{4cm}{\centering\smallskip $f(|P|)$ impossible~under~\sic\smallskip}
        & \cite{Bil13}, 2013 \\\hline
        arbitrary & $O(|E|^\alpha)$, $\alpha \leq \delta$ & 
        \parbox{4cm}{\centering\smallskip $O(|E|^\delta|P|^\beta)$, $\delta + \beta < 2$\\ impossible~under~\ovh \smallskip}
        & \cite{BI16}, 2016 \\\hline
        deterministic DAG & $O(|E|^\alpha)$, $\alpha \leq \delta$ & 
        \parbox{4cm}{\centering\smallskip $O(|E|^\delta|P|^\beta)$, $\delta + \beta < 2$\\ impossible~under~\ovh \smallskip}
        & \cite{Equi18}, 2019 \\\hline
        deterministic DAG & $O(|E|^\alpha)$, $\alpha \in \mathbb{R}$ & 
        \parbox{4cm}{\centering\smallskip $O(|E|^\delta|P|^\beta)$, $\delta + \beta < 2$\\ impossible~under~\ovh \smallskip}
        & This paper \\\hline
        arbitrary & $O(|E|^\alpha)$, $\alpha \in \mathbb{R}$ & 
        \parbox{5cm}{\centering\smallskip $O(|E|^\delta|P|^\beta)$, $\delta < 1$ or $\beta < 1$\\ impossible~under~\ovh \smallskip}
        & This paper \\\hline
    \end{tabular}
    \caption{Upper bounds (first three rows) and conditional lower bounds for offline \smlg on a graph $G = (V,E)$ and a pattern $P$. On the fourth line, $f(\cdot)$ is an arbitrary function.}
    \label{tab:bounds}
\end{table}

\begin{figure}[!ht]
    \begin{subfigure}[t]{0.5\textwidth}
        \centering
        \includegraphics[scale=1.1]{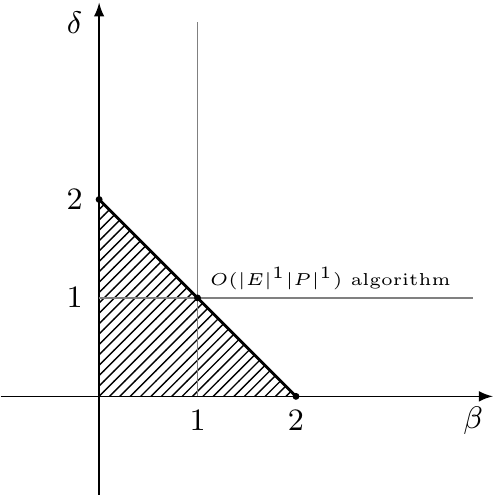}\caption{$\alpha \leq \delta$}
        \label{subfig:AlphaBetaDelta_1}
        \end{subfigure}\hfill
    \begin{subfigure}[t]{0.5\textwidth}
        \centering
        \includegraphics[scale=1.1]{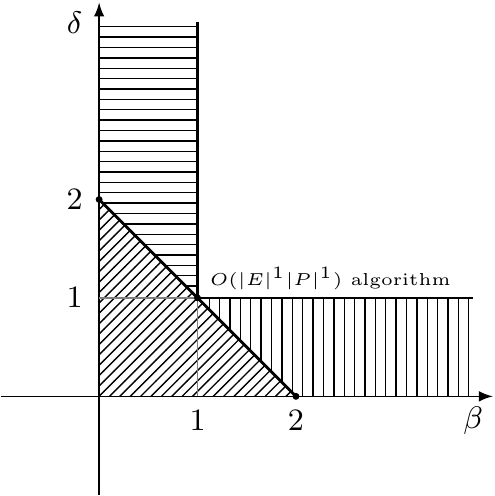}\caption{$\alpha \in \mathbb{R}$}
        \label{subfig:AlphaBetaDelta_2}
    \end{subfigure}
    \caption{ The dashed areas of the plots represent the forbidden values of $\delta$ and $\beta$ for $O(|E|^\delta|P|^\beta)$-time queries, under \ovh. Figure~\ref{subfig:AlphaBetaDelta_1} shows the lower bound that follows from the online case~\cite{BI16,Equi18}, and holds for $\alpha \leq \delta$. Figure~\ref{subfig:AlphaBetaDelta_2} depicts our lower bounds (tight, thanks to the online $O(|E||P|)$-time algorithm from~\cite{AmirLL00}). In addition, these hold for any value of $\alpha$.}
    \label{fig:AlphaBetaDelta}
\end{figure}

\subsection{Results}
\label{section:results}

In the Orthogonal Vectors (\ov) problem we are given two sets $X, Y \subseteq \{ 0,1 \}^d$ such that $|X| = |Y| = N$ and $d = \omega(\log N)$, and we need to decide whether there exists $x \in X$ and $y \in Y$ such that $x$ and $y$ are orthogonal, namely, $x \cdot y = 0$. \ovh states that for any constant $\epsilon > 0$, no algorithm can solve \ov in time $O(N^{2-\epsilon}\text{poly}(d))$. Notice that the better known \emph{Strong Exponential Time Hypothesis} (\seth)~\cite{IP01} implies \ovh~\cite{DBLP:journals/tcs/Williams05}, so all our lower bounds hold also under \seth.

Our results are obtained using a technique used for example in the field of dynamic  algorithms, see e.g.~\cite{DBLP:conf/focs/AbboudW14,DBLP:conf/focs/AbboudRW17}. Recall the reduction from 
$k$-SAT to \ov from~\cite{DBLP:journals/tcs/Williams05}: the $n$ variables of the formula $\phi$ are split into two groups of $n/2$ variables each, all partial $2^{n/2}$ Boolean assignments are generated for each group, and these induce two sets $X$ and $Y$ of size $N = 2^{n/2}$ each, such that \ov returns `yes' on $X$ and $Y$ if and only if $\phi$ is satisfiable. Suppose one could index $X$ to support $O(M^{2-\epsilon}\text{poly}(d))$-time queries for any set $Y$ of $M$ vectors, for some $\epsilon > 0$. One now can adjust the splitting of the variables based on the hypothetical $\epsilon$: the first part (corresponding to $X$) has $n\delta_{\epsilon}$ variables, and the other part (corresponding to $Y$) has $n(1-\delta_{\epsilon})$ variables. We can choose a $\delta_{\epsilon}$ depending on $\epsilon$ such that querying each vector in $Y$ against the index on $X$ takes overall time $O(2^{n(1-\gamma)})$, for some $\gamma > 0$, contradicting \seth. 

In this paper, instead of employing this technique inside the reduction for offline \smlg (as done in previous applications of this technique), we formalize the reason why it works through the notion of a \emph{linear independent-component reduction ($lic$)}. Such a reduction allows to immediately argue that if a problem $A$ is hard to index, and we have a $lic$ reduction from $A$ to $B$, then also $B$ is hard to index (\Cref{lemma:idx_B_implies_idx_A}). Since \ov is hard to index, it follows simply as a corollary that any problem to which \ov reduces is hard to index. In order to get the best possible result for \smlg, we also show that a generalized version of \ov is hard to index (\Cref{theorem:idx-NM-OV-contradicts-OVH}). As such, we upgrade the idea of an ``adjustable splitting'' of the variables from a technique to a directly transferable result, once a $lic$ reduction is shown to exist. 

Examples of problems to which a $lic$ reduction could be applied are those that arise from computing a distance between two elements. Popular examples are edit distance, dynamic time warping distance (DTWD), Frechet distance, longest common subsequence. All these problems have been shown to require quadratic time to be solved under \ovh. The reductions proving these lower bound for DTWD \cite{ABW15} and Frechet distance \cite{Bringmann2014walking} are indeed $lic$ reductions, hence these problems automatically obtain a lower bound also for their offline version. More specifically, \ovh implies that we cannot preprocess the first input of a DTWD or Frechet distance problem in polynomial time and provide sub-quadratic time queries for the second input.

On the other hand, the final sequences used in the reductions for edit distance \cite{BI15} and longest common subsequence \cite{ABW15} present some dependencies within each other, thus they would need to be slightly tweaked to make the definition of $lic$ reduction apply. These cross dependencies only concerns the size of the gadgets used in the reductions and not their structural properties, hence we are confident that the modifications needed to such gadgets require only a marginal effort. 

To easily explain this idea and to better understand the utility of a $lic$ reduction, let us consider edit distance. In a common offline variation of this problem, we are required to build a data structure for a long string $T$ such that one can decide if a given query string $P$ is within edit distance $\kappa$ from a substring of $T$. It suffices to observe that, in the reduction of Backurs and Indyk~\cite{BI15} from \ov to edit distance, this problem is utilized as an intermidiate step, and up to this point the reduction from \ov is indeed a $lic$ reduction (see \Cref{section:self-reducibility}). Hence, we immediately obtain the following result.

\begin{theorem}
\label{theorem:no_idx_pattern}
For any $\alpha>0$, $\beta\geq 1$, and $\delta>0$ such that $\beta+\delta<2$, there is no algorithm preprocessing a string $T$ in time $O(|E|^\alpha)$, such that for any pattern string $P$ we can find a substring of $T$ at minimum edit distance with $P$, in time $O(|T|^\delta|P|^\beta)$, unless \ovh is false.
\end{theorem}

This bound is tight because for $\delta + \beta = 2$ there is a matching online algorithm~\cite{MP80}. \Cref{theorem:no_idx_pattern} also strenghtens the recent result of Cohen{-}Addad, Feuilloley and Starikovskaya~\cite{CFS19}, stating that an index built in polynomial time cannot support queries for approximate string matching in $O(|T|^{\delta})$ time, for any $\delta < 1$, unless \seth is false. However, the boundary condition is different, since in their case $\kappa=O(\log |T|)$, while in our case $\kappa = \Theta(|P|)$.

Our approach for the \smlg problem is similar. In \Cref{section:indexing-labeled-graphs} we revisit the reduction from~\cite{EGMT19} and observe that it is a $lic$ reduction. As such, we can immediately obtain the following result.

\begin{theorem}
\label{theorem:indexing-dags}
For any $\alpha>0$, $\beta\geq 1$, and $\delta>0$ such that $\beta+\delta<2$, there is no algorithm preprocessing a labeled graph $G = (V,E,\ell)$ in time $O(|E|^\alpha)$ such that for any pattern string $P$ we can solve the \smlg problem on $G$ and $P$ in time $O(|E|^\delta|P|^\beta)$, unless \ovh is false.
This holds even if restricted to a binary alphabet, and to deterministic DAGs in which the sum of out-degree and in-degree of any node is at most three.\footnote{We implicitly assumed here that the graph $G$ is the part of the input on which to build the index, because it is the first input to \smlg. However, by exchanging $G$ and $P$, it trivially holds that we also cannot polynomially index a pattern string $P$ to support fast queries in the form of a labeled graph.}
\label{corollary:no_idx_SMLG}
\end{theorem}

This lower bound is tight because for $\delta + \beta = 2$ there is a matching online algorithm~\cite{amir1997pattern}. However, this bound does not disprove a hypothetical polynomial indexing algorithm with query time $O(|E|^{\delta}|P|^{2})$, for some $0 < \delta < 1$. Since graphs in practical applications are much larger than the pattern, such an algorithm would be quite significant. However, when the graph is allowed to have cycles, we also show that this is impossible under \ovh. 

\begin{theorem}
\label{theorem:no_idx_SMLG_cycles}
For any $\alpha>0$, $\beta\geq 1$, and $0 < \delta < 1$, there is no algorithm preprocessing a labeled graph $G = (V,E,\ell)$ in time $O(|E|^\alpha)$ such that for any pattern string $P$ we can solve the \smlg problem on $G$ and $P$ in time $O(|E|^\delta|P|^\beta)$, unless \ovh is false.
\end{theorem}

We obtain \Cref{theorem:no_idx_SMLG_cycles} by slightly modifying the reduction of~\cite{EGMT19} with the introduction of certain cycles, that are necessary to allow for query patterns of length longer than the graph size. We leave as open question whether the lower bound from \Cref{theorem:no_idx_SMLG_cycles} holds also for DAGs. 

\begin{openproblem}
Do there exist $\alpha > 0$, $\beta \geq 1$, $0 < \delta < 1$, and an algorithm preprocessing a labeled (deterministic) DAG $G = (V,E,\ell)$ in time $O(|E|^\alpha)$ such that for any pattern string $P$ we can solve the \smlg problem on $G$ and $P$ in time $O(|E|^\delta|P|^\beta)$?
\end{openproblem}

\section{Formalizing the technique}
\label{sec:definitions}

\subsection{Linear independent-components reductions}
All problems considered in this paper are such that their input is naturally partitioned in two. For a problem $P$, we will denote by $P_X \times P_Y$ the set of all possible inputs for $P$. For a particular input $(p_x,p_y) \in P_X \times P_Y$, we will denote by $|p_x|$ and $|p_y|$ the length of each of $p_x$ and $p_y$, respectively. Intuitively, $p_x$ represents what we want to build the index on, while $p_y$ is what we want to query for. We start by formalizing the concept of \emph{indexability}.

\begin{definition}[Indexability]
\label{def:indexable}
Problem $P$ is ($I$,$Q$)\emph{-indexable} if for every $p_x \in P_X$ we can preprocess $p_x$ in time $I(|p_x|)$ such that for every $p_y \in P_Y$ we can solve $P$ on $(p_x,p_y)$ in time $Q(|p_x|,|p_y|)$.
\end{definition}

We further refine this notion into that of \emph{polynomial indexability}, by specifying the degree of the polynomial costs of building the index and of performing the queries.

\begin{definition}[Polynomial indexability]
\label{def:poly_indexable_param}
Problem $P$ is ($\alpha, \delta, \beta$)\emph{-polynomially indexable with parameter $k$} if $P$ is ($I$,$Q$)-indexable and $I(|p_x|) = O(k^{O(1)}|p_x|^\alpha)$ and $Q(k^{O(1)}|p_x|,|p_y|)=O(|p_x|^\delta|p_y|^\beta)$. If $k = O(1)$, then we say that $P$ is ($\alpha, \delta, \beta$)\emph{-polynomially indexable}.
\end{definition}

The introduction of parameter $k$ is needed to be consistent with \ovh, since when proving a lower bound conditioned on \ovh, the reduction is allowed to be polynomial in the vector dimension $d$. As we will see, we will set $k=d$.

We now introduce linear independent-components reductions, which we show below in \Cref{lemma:idx_B_implies_idx_A} to maintain $(\alpha,\delta,\beta)$-polynomial indexability.

\begin{definition}[$lic$ reduction]
\label{def:independent-components_reduction}
Problem $A$ has a \emph{linear independent-components ($lic$) reduction with parameter $k$} to problem $B$, indicated as $A \leq_{lic}^k B$, if the following two properties hold:
\begin{itemize}
    \item[i)] \textbf{Correctness}: There exists a reduction from $A$ to $B$ modeled by functions $r_x$, $r_y$ and $s$. That is, for any input $(a_x,a_y)$ for $A$, we have $r_x(a_x) = b_x$, $r_y(a_y) = b_y$, $(b_x, b_y)$ is a valid input for $B$, and $s$ solves $A$ given the output $B(b_x, b_y)$ of an oracle to $B$, namely $s(B(r(a_x),r(a_y)))=A(a_x,a_y)$.
    \item[ii)] \textbf{Parameterized linearity}: Functions $r_x$, $r_y$ and $s$ can be computed in linear time in the size of their input, multiplied by $k^{O(1)}$.
\end{itemize}
\end{definition}

\begin{lemma}
\label{lemma:idx_B_implies_idx_A}
Given problems $A$ and $B$ and constants $\alpha>0, \delta>0, \beta\geq1$, if $A \leq_{lic}^k B$ holds, and $B$ is ($\alpha, \delta, \beta$)-polynomially indexable, then $A$ is ($\alpha, \delta, \beta$)-polynomially indexable with parameter~$k$.
\end{lemma}
\begin{proof}
Let $a_x \in A_X$ be the first input of problem $A$. The linear independent-components reduction computes the first input of problem $B$ as $b_x = r_x(a_x)$ in time $O(k^{O(1)}|a_x|)$. This means that $|b_x| = O(k^{O(1)}|a_x|)$, since the size of the data structure that we build with the reduction cannot be greater than the time spent for performing the reduction itself. Problem $B$ is ($\alpha, \delta, \beta$)-polynomially indexable, hence we can build an index on $b_x$ in time $O(|b_x|^\alpha)$ in such a way that we can perform queries for every $b_y$ in time $O(|b_x|^\delta|b_y|^\beta)$. Now given any input $a_y$ for $A$ we can compute its corresponding $b_y = r_y(a_y)$ via the reduction in time $O(k^{O(1)}|a_y|)$ and answer a query for it using the index that we built on $b_x$. Again, notice that $|b_y| = O(k^{O(1)}|a_y|)$. The cost for such a query is $O(k^{O(1)}|a_y| + |b_x|^\delta|b_y|^\beta) = O(k^{O(1)}|a_y| + k^{O(1)}|a_x|^\delta |a_y|^\beta)$, which, since $\delta>0$, is the same as $O(k^{O(1)}|a_x|^\delta |a_y|^\beta)$ when $\beta\geq1$. Notice that the indexing time is $O(|b_x|^\alpha) = O(k^{O(1)}|a_x|^\alpha)$. Hence $A$ is ($\alpha, \delta, \beta$)-polynomially indexable with parameter $k$, when $\beta\geq1$.
\end{proof}

\subsection{Conditional indexing lower bounds}
\label{section:self-reducibility}

We begin by stating, with our formalism, a known strengthening of the hardness of indexing reduction presented at the beginning of \Cref{section:results} (note that it also follows as a special case of \Cref{theorem:idx-NM-OV-contradicts-OVH} below).

\begin{theorem}[Folklore]
\label{theorem:idx_standard_OV_contradicts_OVH}
If \ov is ($\alpha, \delta, \beta$)-polynomially indexable with parameter $d$, and $\beta + \delta < 2$, then \ovh fails.
\end{theorem}

The value of a parameterized $lic$ reduction can now be apprehended: once a parameterized $lic$ reduction is shown to exist, the indexing lower bound follows directly.
\begin{corollary}
\label{corollary:no_idx_P_if_standard_OV_reduction}
Any problem $P$ such that $\ov \leq_{lic}^d P$ holds is not ($\alpha, \delta, \beta$)-polynomially indexable, for any $\alpha>0$, $\beta\geq 1$, $\delta > 0$ with $\beta + \delta < 2$, unless \ovh is false.
\end{corollary}

\begin{proof}
Assume by contradiction that $P$ is ($\alpha, \delta, \beta$)-polynomially indexable. Apply \Cref{lemma:idx_B_implies_idx_A} to prove that \ov is ($\alpha, \delta, \beta$)-polynomially indexable with parameter $d$, and $\beta + \delta < 2$; this contradicts \Cref{theorem:idx_standard_OV_contradicts_OVH}.
\end{proof}

For a simple and concrete application of \Cref{corollary:no_idx_P_if_standard_OV_reduction}, consider the following problem, in which $d(S_1,S_2)$ denotes the edit distance between strings $S_1$ and $S_2$.

\begin{problem}[\pattern]
\item{\textsc{Input}:} Two strings $T$ and $P$.
\item{\textsc{Output}:} $\displaystyle \min_{\text{$S$ substring of $T$}} d(S,P)$.
\end{problem}

Backurs and Indyk~\cite{BI15} reduce \ov to \pattern by constructing a string $T$ based solely on the first input $X$ to \ov and a string $P$ based solely on the second input $Y$ to \ov, such that if there are two orthogonal vectors then the answer to \pattern on $T$ and $P$ is below a certain value, and if there are not, then the answer is equal to another specific value. Each of $T$ and $P$ can be constructed in time $O(d^{O(1)}N) = O(d^{O(1)}(dN))$. This is a \emph{lic} reduction with parameter $d$. Directly applying \Cref{corollary:no_idx_P_if_standard_OV_reduction}, we obtain \Cref{theorem:no_idx_pattern}.

\subsection{Indexing Generalized Orthogonal Vectors}

\Cref{corollary:no_idx_P_if_standard_OV_reduction} will suffice to prove \Cref{theorem:indexing-dags}. However, in order to prove that no query time $O(|E|^{\delta}|P|^{\beta})$ is possible for any $\delta < 1$, we need a strengthening of \Cref{theorem:idx_standard_OV_contradicts_OVH}. As such, we introduce the generalized $(N,M)$-\emph{Orthogonal Vectors} problem, as follows:
\begin{problem}[$(N,M)$-\ov]
\item{\textsc{Input}:} Two sets $X,Y \subseteq \{0,1\}^d$, such that $|X|=N$ and $|Y|=M$.
\item{\textsc{Output}:} $True$ if and only if there exists $(x,y) \in X \times Y$ such that $x \cdot y = 0$.
\end{problem}

The theorem below is the desired generalization of \Cref{theorem:idx_standard_OV_contradicts_OVH}, since it implies, for example, that we cannot have $O(N^{1/2}M^2)$-time queries after polynomial-time indexing. To the best of our efforts, we could not find a proof of this result in the literature, and hence we give one here. It is based on the same idea of an ``adjustable splitting'' into subvectors, a part of which is indexed, while the other part is queried. However, some technical subtleties arise from the combination of all parameters $\alpha,\delta,\beta$.

\begin{theorem}
\label{theorem:idx-NM-OV-contradicts-OVH}
If $(N,M)$-\ov is ($\alpha, \delta, \beta$)-polynomially indexable with parameter $d$, and either $\delta < 1$ or $\beta < 1$, then \ovh fails. That is, under \ovh, we cannot support $O(N^\delta M^\beta)$-time queries for $(N,M)$-\ov, for either $\delta < 1$ or $\beta < 1$, even after polynomial-time preprocessing.
\end{theorem}
\begin{proof}
Let $X$ and $Y$ be the input for \ov and assume that their length is $n$. Our strategy is to split this instance of \ov into many $(N,M)$-\ov sub-problems and show that a too efficient indexing scheme for $(N,M)$-\ov applied to such sub-problems would lead to an online algorithm for \ov running in sub-quadratic time, hence contradicting \ovh. The key is to adjust the size of such $(N,M)$-\ov sub-problems to fit our needs. Let us begin by partitioning set $X$ into subsets of $N$ vectors each, and set $Y$ into subsets of $M$ vectors each, as shown in Figure~\ref{fig:Vector_Sets}.
\begin{figure}[ht]
    \centering
    \includegraphics[scale=.8]{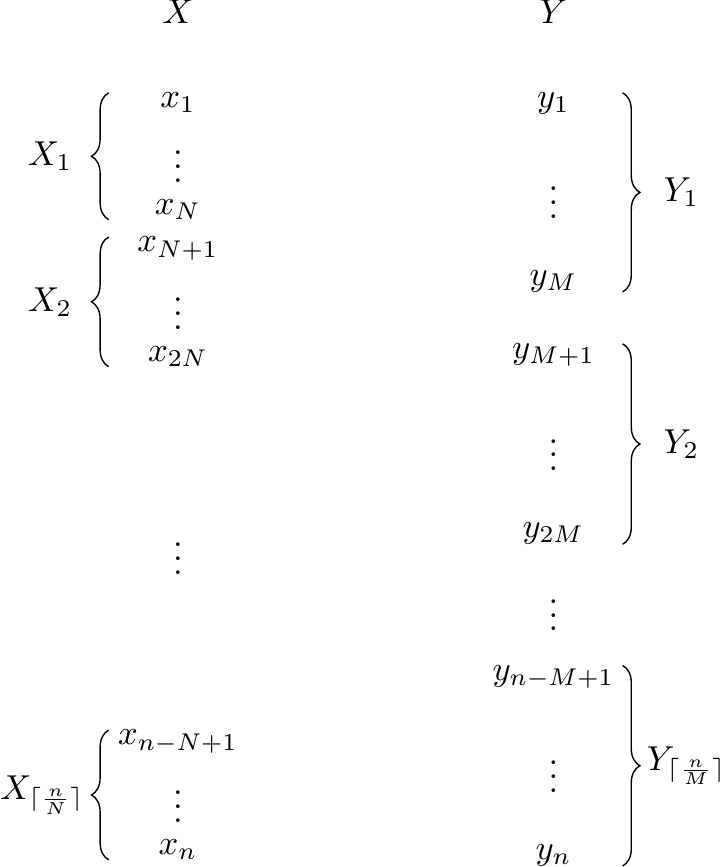}
    \caption{Two sets of $n$ vectors and their partitioning into sub-sets of $X$ for indexing, and sub-sets of $Y$ for querying.}
    \label{fig:Vector_Sets}
\end{figure} The instances of $(N,M)$-\ov sub-problems that we want to consider are all those pairs of vector sets $(X_i, Y_j)$ in which $X_i$ is a subset of $X$ and $Y_j$ is a subset of $Y$. Solving all the $(X_i, Y_j)$ instances solves the original problem.\footnote{The idea of splitting the two sets into smaller groups was also used in~\cite{AWY15} to obtain a fast randomized algorithm for \ov, based on the polynomial method, and therein the groups always had equal size.} Now we proceed to index sub-sets $X_i$ and to analyze how the time complexity of the original problem looks like when expressed in terms of the $(N,M)$-\ov sub-problems. Later we show how we can obtain a sub-quadratic time algorithm for \ov by choosing specific values for $N$ and $M$.

Since we are assuming that $(N,M)$-\ov is ($\alpha, \delta, \beta$)-polynomially indexable with parameter $d$, we can build index $Idx(X_i)$ for subset $X_i$ of $N$ vectors in time $O(d^{O(1)}(d N)^\alpha)$, and additionally we can answer a query for any subset $Y_j$ of $M$ vectors using index $Idx(X_i)$ in time $O(d^{O(1)}(d N)^\delta (d M)^\beta)$. Hence, given index $Idx(X_i)$, we can solve sub-problems $(X_i, Y_j)$ for a fixed $i$ and $j$ ($1\leq j \leq \lceil \frac n M \rceil$) by performing $\lceil\frac n M\rceil$ queries, one for each subset $Y_j$ of $Y$. Repeating this process for all $X_i$ covers all possible pairs $(X_i, Y_j)$, and since we have $\lceil\frac n N \rceil \lceil \frac n M\rceil $ such pairs, the total cost for solving \ov is:
\begin{equation}
O\left(d^{O(1)}(d N)^\alpha \frac n N + d^{O(1)}(d N)^\delta (d M)^\beta \frac n N \frac n M\right) =
O\left(d^{O(1)}\left(N^{\alpha-1} n + N^{\delta-1}M^{\beta-1}n^2\right)\right).
\label{eq:time_complexity_1}
\end{equation}
In order to have a sub-quadratic-time algorithm for \ov we need both of the terms of the sum above to be sub-quadratic. Namely, our time complexity should be $O\left(d^{O(1)}\left(n^{2-\epsilon'}+n^{2-\epsilon}\right)\right)$, for some $\epsilon, \, \epsilon'> 0$. Notice that in order to prove \ovh to be wrong it is enough to find one specific value for $\epsilon$ and one for $\epsilon'$ such that the following two conditions hold:
\begin{align*}
&\text{(a) : } N^{\alpha-1}n = O(n^{2-\epsilon'})\\
&\text{(b) : } N^{\delta-1}M^{\beta-1}n^2 = O(n^{2-\epsilon})
\end{align*}
As a first observation, notice that we need also to enforce $1\leq N \leq n$ and $1\leq M \leq n$. This is because every $X_i$ and every $Y_j$ must contain at least one vector in order to be an instance of $(N,M)$-\ov, and trivially their size must not exceed the size $n$ of the original \ov instance. Moreover, $N$ and $M$ must also be integers. This last requirement might cause some complications during our analysis, and due to this reason we will take advantage of a useful trick. We introduce new variables $\tilde{N}$ and $\tilde{M}$ so that we can make them assume real values. Our actual $N$ and $M$ would be the ceiling of $\tilde{N}$ and $\tilde{M}$. Putting all together, we want that for every $n\in \mathbb{N}, \alpha, \delta, \beta > 0$ such that either $\delta<1$ or $\beta<1$ there exists $\epsilon>0, \epsilon'>0, N, M, \tilde{N}, \tilde{M}$ such that:
\begin{align*}
(\text{a}) \; &N^{\alpha-1}n = O(n^{2-\epsilon'})\\
(\text{b}) \; &N^{\delta-1}M^{\beta-1}n^2 = O(n^{2-\epsilon})\\
(\atilde) \; &\tilde{N}^{\alpha-1}n = n^{2-\epsilon'}\\
(\btilde) \; &\tilde{N}^{\delta-1}\tilde{M}^{\beta-1}n^2 = n^{2-\epsilon}\\
(\text{c}) \; &N = \lceil \tilde{N} \rceil, \; M = \lceil \tilde{M} \rceil\\
(\text{d}) \; &1\leq \tilde{N} \leq n, \; 1\leq \tilde{M} \leq n
\end{align*}
Notice that forcing $1\leq \tilde{N} \leq n$ also ensures $1\leq N \leq n$, since we are taking the ceiling $N = \lceil \tilde{N} \rceil$. The same holds for $\tilde{M}$ and $M$.

We start our case analysis by identifying two cases for parameter $\alpha$, namely $\alpha \neq 1$ and $\alpha = 1$. These are eventually broken down into specific sub-cases for parameters $\delta$ and $\beta$. The strategy is to prove that if conditions ($\atilde$), ($\btilde$), (c), (d), (e) hold, then also conditions (a) and (b) hold. For simplicity, we report here only the most interesting cases in which $\alpha\neq$ 0, $\delta \neq 1$ and $\beta\neq 1$. The complete analysis of the remaining cases can be found in \Cref{appendix:missing-cases}.

\textbf{Case 1}: $\alpha \neq 1$. In this case we obtain the following constraint on $\tilde{N}$ from condition ($\atilde$):
\begin{equation}
\label{eq:constraint_b2-1}
\tilde{N}^{\alpha-1} n = n^{2-\epsilon'} \Leftrightarrow \tilde{N} = n^{\frac{1-\epsilon'}{\alpha-1}}.
\end{equation}
Now, given $\epsilon'$, we can compute $\tilde{N}$ using this equation so that we satisfy condition ($\atilde$). In doing so we need to make sure that condition (d) is also respected. To this end, we need to check that $\epsilon'$ satisfies $0 \leq \frac{1-\epsilon'}{\alpha-1} \leq 1$. Let us start with the left inequality.
\begin{align}
\label{eq:constraint_b2-2}
\frac{1-\epsilon'}{\alpha-1} \geq 0 &\Leftrightarrow (1-\epsilon'\geq0 \text{ and } \alpha-1>0) \text{ or } (1-\epsilon'\leq0 \text{ and } \alpha-1<0)\nonumber\\
&\Leftrightarrow (\epsilon'\leq1 \text{ and } \alpha>1) \text{ or } (\epsilon'\geq1 \text{ and } \alpha<1)
\end{align}
For the right inequality, we first handle the case in which $\epsilon'\leq1$ and $\alpha>1$ and we combine it with the constraint $\frac{1-\epsilon'}{\alpha-1} \leq 1$. Since $\alpha>1$ then $\alpha-1>0$ and we have that $\frac{1-\epsilon'}{\alpha-1} \leq 1 \Leftrightarrow \epsilon' \geq 2-\alpha$. Hence, the final constraint for $\epsilon'$ is $2-\alpha \leq \epsilon' \leq 1$, and we know that there exists valid values for $\epsilon'$ since $\alpha>1 \Rightarrow 2-\alpha < 1$.

Now we take into account the other case, namely $\epsilon' \geq 1$ and $\alpha<1$. We find ourselves in a symmetric situation in which $\alpha<1$ implies $\alpha-1<0$ which leads to $\frac{1-\epsilon'}{\alpha-1} \leq 1 \Leftrightarrow \epsilon' \leq 2-\alpha$. Putting all together we have $1 \leq \epsilon' \leq 2-\alpha$, and the existence of valid values for $\epsilon'$ is guaranteed by the fact that $\alpha<1 \Rightarrow 2-\alpha > 1$.

So far we analyzed conditions ($\atilde$) and (d), for $\tilde{N}$. To analyze the other conditions, we need to consider three sub-cases. Here we present the more challenging one, which is in turn split into two more sub-cases. The reader can find the others in \Cref{appendix:missing-cases}.

\textbf{Case 1.1}: $\delta\neq1$ and $\beta\neq1$. Now condition ($\btilde$) yields the following:
\begin{equation}
\label{eq:constraint_a3-1}
\tilde{N}^{\delta-1}\tilde{M}^{\beta-1}n^2=n^{2-\epsilon} \Leftrightarrow \tilde{M} = \tilde{N}^\frac{1-\delta}{\beta-1}n^\frac{\epsilon}{1-\beta}.
\end{equation}
We apply the substitution $\tilde{N}=n^{\frac{1-\epsilon'}{\alpha-1}}$ that we obtained from equation (\ref{eq:constraint_b2-1}). Hence:
\[
\tilde{M} = n^{\frac{1-\epsilon'}{\alpha-1}\frac{1-\delta}{\beta-1}}n^\frac{\epsilon}{1-\beta} = n^{\frac{1-\epsilon'}{\alpha-1}\frac{1-\delta}{\beta-1}-\frac{\epsilon}{\beta-1}}.
\]
We apply condition (d) obtaining the following constraint:
\begin{equation}
\label{eq:constraint_a3-2}
0 \leq \frac{1-\epsilon'}{\alpha-1}\frac{1-\delta}{\beta-1}-\frac{\epsilon}{\beta-1} \leq 1
\end{equation}
Here we face two more sub-cases.

\textbf{Case 1.1.1}: $\beta-1<0 \Leftrightarrow \beta<1$. We extract the constraint on $\epsilon$ from the two inequalities in \eqref{eq:constraint_a3-2} above. We start by analysing the left inequality.
\begin{align*}
\frac{1-\epsilon'}{\alpha-1}\frac{1-\delta}{\beta-1}-\frac{\epsilon}{\beta-1} &\geq 0\\
\epsilon &\geq\frac{1-\epsilon'}{\alpha-1}(1-\delta).
\end{align*}
From the second inequality instead we get:
\begin{align*}
\frac{1-\epsilon'}{\alpha-1}\frac{1-\delta}{\beta-1}-\frac{\epsilon}{\beta-1} &\leq 1\\
\epsilon-\frac{1-\epsilon'}{\alpha-1}(1-\delta) &\leq 1-\beta\\
\epsilon &\leq 1-\beta + \frac{1-\epsilon'}{\alpha-1}(1-\delta)
\end{align*}
Since $\epsilon>0$, we need to be sure that the right term of this last inequality is strictly greater than $0$. Since $1-\beta>0$ and $\frac{1-\epsilon'}{\alpha-1}>0$ (from \eqref{eq:constraint_b2-2}), the interesting case is when $1-\delta<0$. Notice that $\frac{1-\epsilon'}{\alpha-1} \rightarrow 0$ as $\epsilon' \rightarrow 1$. This means that we can always choose $\epsilon'$ as close to $1$ as needed to make $1-\beta + \frac{1-\epsilon'}{\alpha-1}(1-\delta)>0$ hold.

At this point we have proved conditions ($\atilde$), ($\btilde$) and (d), thus we are left to show that conditions (a) and (b) also hold. To this end, let us choose $\epsilon$, $\epsilon'$, $\tilde{N}$ and $\tilde{M}$ in such a way that conditions ($\atilde$), ($\btilde$) and (d) are verified. Then we choose $N = \lceil N \rceil$ and $M = \lceil M \rceil$ so that condition (c) is verified. We analyse in depth condition (b) since it is more complicated; condition (a) can be proven applying the same technique. We first remark the following property:

\begin{fact}
\label{fact:ceil_to_bigO}
$\forall n,a,b\in\mathbb{R}.\;\lceil n^a \rceil^b = O(n^{ab})$.
\end{fact}
\begin{proof}
For $b=0$ the statement is trivially true. If $b>0$, we have $\lceil n^a \rceil^b \leq (n^a+1)^b = O(n^{ab})$. If $b<0$, we have $\lceil n^a \rceil^b \leq (n^a-1)^b = O(n^{ab})$.
\end{proof}

Now we can show that
\begin{align*}
N^{\delta-1}M^{\beta-1}n^2 &= \lceil n^\frac{1-\epsilon'}{\alpha-1}\rceil^{\delta-1}\lceil n^{\frac{1-\epsilon'}{\alpha-1}\frac{1-\delta}{\beta-1}-\frac{\epsilon}{\beta-1}}\rceil^{\beta-1}n^2\\
&=O\left( n^{-\frac{1-\epsilon'}{\alpha-1}(1-\delta)} n^{\frac{1-\epsilon'}{\alpha-1}(1-\delta)-\epsilon}n^2\right)\\
&= O(n^{2-\epsilon})
\end{align*} where the first step is justified by Fact~\ref{fact:ceil_to_bigO}. We conclude that both conditions (a) and (b) hold.

\textbf{Case 1.1.2}: $\beta-1>0 \Leftrightarrow \beta>1$. This case is symmetric to the previous one and implies that we are in the situation $\delta<1$. When  extracting the constraints on $\epsilon$, we will have the same inequalities but with the opposite direction. Indeed, we multiply by $\beta-1$ which now has the opposite sign.
\[
1-\beta + \frac{1-\epsilon'}{\alpha-1}(1-\delta) \leq \epsilon \leq \frac{1-\epsilon'}{\alpha-1}(1-\delta)
\]
Given that $\epsilon>0$, we need to verify to have room to choose such an $\epsilon$, that is $\frac{1-\epsilon'}{\alpha-1}(1-\delta) > 0$. We know that the first factor of this multiplication is between $0$ and $1$. Hence, we can always choose an $\epsilon'$ such that $\frac{1-\epsilon'}{\alpha-1} > 0$. Moreover, in this sub-case we have $\delta<1$ which ensures that also $1-\delta$ is strictly positive. Hence, the quantity $\frac{1-\epsilon'}{\alpha-1}(1-\delta)$ is strictly positive, which means that there always exists an $\epsilon$ such that condition (d) holds.
Assuming condition (c), conditions (a) and (b) can be proved to hold in the same manner as in the previous sub-case.

In conclusion, we can say that depending on $\alpha$, $\delta$ and $\beta$ we find ourselves into one of the listed cases. We showed that in each one of those we can always find values for $\epsilon$ and $\epsilon'$ such that there exists integer values for $N$ and $M$ that provide an algorithm for \ov running in time $O(n^{2-\epsilon} + n^{2-\epsilon'})$, proving \ovh to be false.
\end{proof}

\begin{corollary}
\label{corollary:no_idx_P_if_NM-OV_reduction}
Any problem $P$ such that $(N,M)$-$OV \leq_{lic}^d P$ holds is not ($\alpha, \delta, \beta$)-polynomially indexable, for any $\alpha>0$, $\beta\geq 1$, $0<\delta<1$, unless \ovh is false.
\end{corollary}

\section{Indexing Labeled Graphs for String Matching}
\label{section:indexing-labeled-graphs}

Recall the following conditional lower bound for \smlg from Equi et al.~\cite{EGMT19}.

\begin{theorem}[\cite{EGMT19}] 
\label{theorem:Emlowerbound}
\begin{sloppypar}
For any $\epsilon > 0$, \smlg on labeled deterministic DAGs cannot be solved in either $O(|E|^{1-\epsilon} \, |P|)$ or $O(|E| \, |P|^{1-\epsilon})$ time unless \ovh fails. This holds even if restricted to a binary alphabet, and to DAGs in which the sum of out-degree and in-degree of any node is at most three.
\end{sloppypar}
\end{theorem}

Given an \ov instance with sets $X$ and $Y$, the reduction from~\cite{EGMT19} builds a graph $G$ using solely~$X$, and a pattern $P$ using solely~$Y$, both in linear time $O(dN)$, such that $P$ has a match in $G$ if and only if there exists a pair of orthogonal vectors.\footnote{Notice that~\cite{EGMT19} originally built $P$ based on $X$, and $G$ based on $Y$. Since it is immaterial for correctness, and in order to keep in line with the notation in this paper, we assumed the opposite here.} This shows that the two conditions of the linear independent-components reduction property hold, thus $\ov \leq_{lic}^d \smlg$. Directly applying \Cref{corollary:no_idx_P_if_standard_OV_reduction}, we obtain \Cref{theorem:indexing-dags}.

Next, we show that constraint $\beta+\delta<2$ can be dropped from \Cref{corollary:no_idx_SMLG} when we are indexing non-deterministic graphs with cycles. The idea is that if we allow $(N,M)$-\ov instances with $M > N$, then the reduction from~\cite{EGMT19} no longer holds, because the pattern $P$ is too large to fit inside the DAG $G$. As such, we need to make a minor adjustment to $G$. For this, we must give some additional details of that reduction. For our purposes, it is enough to explain the construction of a non-deterministic graph from~\cite[Section~2.3]{EGMT19}.

Pattern $P$ is over the alphabet $\Sigma = \{ \B,\E,\zero,\one \}$, has length $|P| = O(dM)$, and can be built in $O(dM)$ time from the second set of vectors $Y = \{ y_1, \ldots, y_M\}$. Namely, we define \[P = \B\B P_{y_1}\E\,\B P_{y_2}\E \ldots \B P_{y_M}\E\E\] where $P_{y_i}$ is a string of length $d$ that is associated with each $y_i \in Y$, for $1 \leq i \leq M$.
The $h$-th symbol of $P_{y_i}$ is either \zero or \one, for each $h \in \{1,\dots,d\}$, such that $P_{y_i}[h] = \one$ if and only if $x_i[h] = 1$.

Starting from the first set of vectors $X$, we define the directed graph $G_W = (V_W,E_W,L_W)$, which can be built in $O(dN)$ time and consists of $N$ connected components $G_W^{(j)}$, one for each vector $x_j \in X$. Component $G_W^{(j)}$ can be constructed so that it contains an occurrence of a subpattern $P_{y_i}$ if and only if $x_j \cdot y_i = 0$. In addition, we need a universal gadget $G_U = (V_U,E_U,L_U)$ of $2N-2$ components $G_{U1}^{(k)}$, where each component can match any of the subpatterns $P_{y_i}$. We actually need two copies $G_{U1}$ and $G_{U2}$ of such universal gadgets, a ``top'' one, and a ``bottom'' one, respectively. All the gadgets are connected as indicated in \figurename~\ref{fig:G_nondet} and the resulting graph $G$ has total size $O(dN)$.

\begin{figure}[t]
    \centering
    \includegraphics[width=\textwidth]{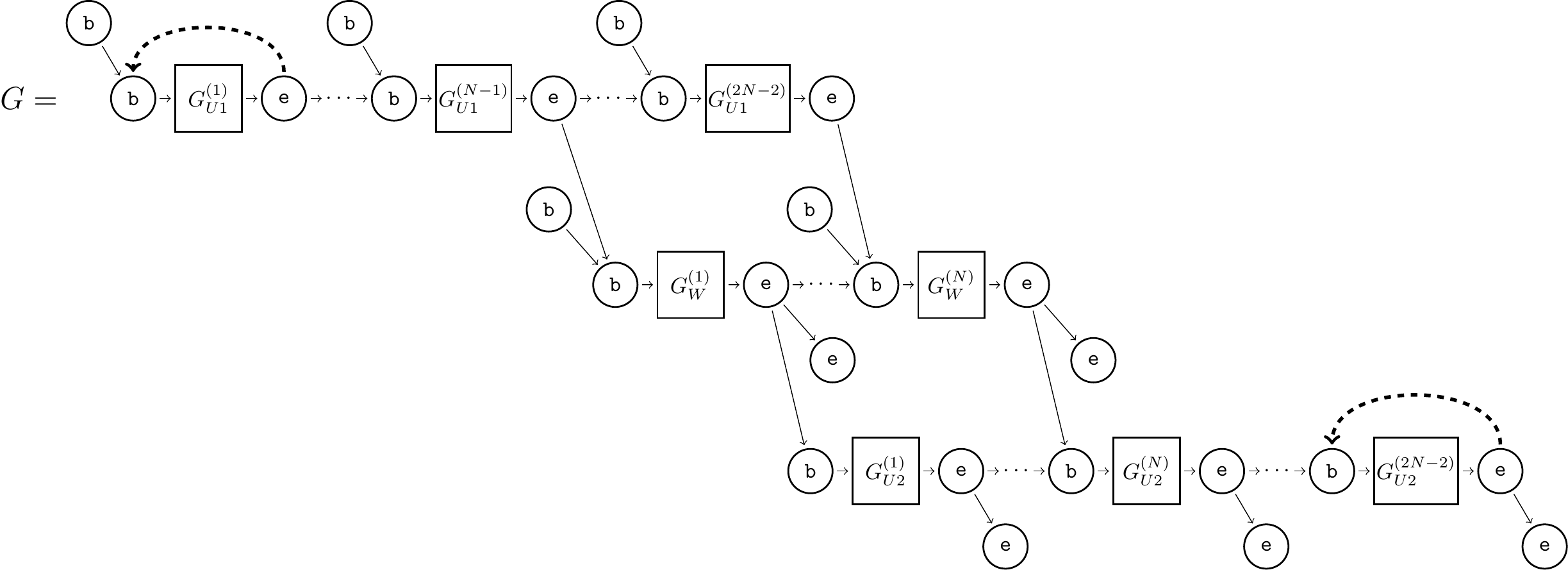}
    \caption{Non-deterministic graph $G$. The dashed thick edges are not present in the acyclic graph from~\cite{EGMT19}, and must be added to handle $(N,M)$-\ov instances with $M > N$.}
    \label{fig:G_nondet}
\end{figure}

The intuition is that a prefix of $P$ is handled by the ``top'' universal gadgets $G_{U1}$, a possible matching a subpattern $P_{y_i}$ of $P$ by one of the ``middle'' gadgets $G_W^{(j)}$, and a suffix of $P$ by the ``bottom'' universal gadgets, because $P$ has a $\B\B$ prefix and an $\E\E$ suffix. As mentioned above, by doing this we cannot accommodate $(N,M)$-\ov instances with $M > N$. However, we can easily fix this by adding a cycle in each of the ``top'' and ``bottom'' universal gadgets, so that a longer pattern will match a universal gadget in this cycle as many times needed to fit inside the graph. More precisely, we can add an edge from the \E-node to the right of $G_{U1}^{(1)}$ back to the \B-node to the left of $G_{U1}^{(1)}$, and likewise from the \E-node to the right of $G_{U2}^{(2N-2)}$ back to the \B-node to the left of $G_{U2}^{(2N-2)}$ (see \figurename~\ref{fig:G_nondet}). It can easily be checked that it still holds that $P$ has a match in the resulting graph $G$ if and only if there are two orthogonal vectors, no matter the relationship between $N$ and $M$. Applying \Cref{corollary:no_idx_P_if_NM-OV_reduction}, we obtain \Cref{theorem:no_idx_SMLG_cycles}.

\bibliography{references}

\newpage
\appendix

\section{Connection to \sic\label{app:sic}}

Given sets $S^1, S^2, \ldots, S^n \subseteq [1..u]$, where $u=\log^c n$ for sufficiently large $c$, the \emph{Set Intersection Conjecture (\sic)} \cite{PR14} is that there is no index of size $O(n^{2-\epsilon})$ for any $\epsilon>0$ to answer in constant time if two sets $S^i$ and $S^j$ intersect or not (i.e. there is no improvement over the table of all precomputed solutions). The reduction of \cite{Bil13} is as follows: build a simple DAG with one copy of the sets as source nodes and another copy as sink nodes. Then add nodes in between corresponding to the elements of the sets. Connect source node corresponding to $S^i$ to all nodes corresponding to elements $v \in S^i$, and all nodes corresponding to $v \in S^i$ to the sink corresponding to $S^i$, for all $i$. Label sources and sinks with their set identifier, and nodes in between with some common letter, say \texttt{A}. Since the graph size is $O(n \log^c n)$, a truly sub-quadratic size index supporting string queries of the form $P=i\mathtt{A}j$ even, say, in exponential time in $|P|$ would prove \sic false. Modifying the relationship between universe size $u$ and number of sets $n$ gives rise to several refined lower bounds for the tradeoff betweeen index construction and query time \cite{GLP19}, which directly transfer to the graph indexing problem through the simple connection stated above.

\section{Missing cases of the proof of \Cref{theorem:idx-NM-OV-contradicts-OVH}}
\label{appendix:missing-cases}

\textbf{Case 2}: $\alpha = 1$. Condition ($\atilde$) simply becomes $n=n^{2-\epsilon'}$, which is verified for $\epsilon'=1$. We now split the analysis of condition ($\btilde$) into two sub-cases.

\textbf{Case 2.1}: $\delta < 1$ and no constraint on $\beta$. We can rewrite condition ($\btilde$) as:
\[
\tilde{N} = \tilde{M}^\frac{1-\beta}{\delta-1}n^\frac{\epsilon}{1-\delta}
\]
since $\delta<1$ guarantees $\delta - 1\neq 0$. If we choose $\tilde{M}=1$ we respect condition (d) and we obtain $\tilde{N} = n^\frac{\epsilon}{1-\delta}$ for any value of $\beta$. Hence, we can first choose a value for $\epsilon$ and later use this equation to obtain the right value for $\tilde{N}$ that will satisfy condition ($\btilde$). Nevertheless, we cannot just pick any value for $\epsilon$. Indeed, we need to guarantee also that condition (d) is holding. This can be achieved by verifying that $0\leq \frac{\epsilon}{1-\delta} \leq 1$. Since $\delta<1$ and $\epsilon>0$ we know that $\frac{\epsilon}{1-\delta} > 0$. Moreover, $\frac{\epsilon}{1-\delta} \leq 1 \Leftrightarrow \epsilon \leq 1-\delta$, which means that any $\epsilon$ such that $0 < \epsilon \leq 1-\delta$ satisfies condition (d). We know that there exists such an $\epsilon$ since $1-\delta>0$.

We are now left to prove that conditions (a) and (b) hold. We proceed as in case 1.1.1 by assuming condition (c) and proving conditions (a) and (b). Condition (a) is easily verified since $\alpha=1$. Since $N = \lceil \tilde{N} \rceil= \lceil n^\frac{\epsilon}{1-\delta} \rceil$ and $M=\lceil \tilde{M} \rceil=1$, and noticing that $\delta-1<0$, we can analyse condition (b) as follows.
\begin{align*}
N^{\delta-1}M^{\beta-1}n^2 &= \lceil n^\frac{\epsilon}{1-\delta} \rceil^{\delta-1}n^2\\
&\leq \left(n^\frac{\epsilon}{1-\delta} - 1 \right)^{\delta-1}\cdot n^2\\
&=O(n^{\frac{\epsilon}{1-\delta}{\delta-1}}n^2)\\
&=O(n^{2-\epsilon}).
\end{align*}
Hence, condition (b) is verified and so all the conditions hold.

\textbf{Case 2.2}: $\beta < 1$ and no constraint on $\delta$. This case is symmetric to the previous one. Indeed, we now rewrite condition ($\btilde$) as:
\[
\tilde{M} = \tilde{N}^\frac{1-\delta}{\beta-1}n^\frac{\epsilon}{1-\beta}
\]
where $\beta<1$ gives $\beta - 1\neq 0$. This time we choose $\tilde{N}=1$, from which we obtain $\tilde{M} = n^\frac{\epsilon}{1-\beta}$ for any value of $\delta$. Again, we will use this equation to find the right value for $\tilde{N}$ once we have chosen $\epsilon$. When choosing such $\epsilon$, we will have to respect the constraint $0 \leq \frac{\epsilon}{1-\beta}\leq 1$ in order to make condition (d) hold. Hence any $\epsilon$ such that $0 < \epsilon \leq 1-\beta$ satisfies condition (d), and $\beta < 1$ guarantees that such an $\epsilon$ exists.

As in the previous case, condition (a) is easily verified by $\alpha=1$. For verifying condition (b) we choose $\epsilon, \epsilon', \tilde{N}, \tilde{M}$ such that conditions ($\atilde$), ($\btilde$) and (d) are verified. Then we choose $N = \lceil \tilde{N} \rceil = 1$ and $M = \lceil \tilde{M} \rceil = \lceil n^\frac{\epsilon}{1-\beta} \rceil$ so that condition (c) is verified. The analysis of condition~(b) is analogous to the previous case and yields $N^{\delta-1}M^{\beta-1}n^2 \leq \left(n^\frac{\epsilon}{1-\beta} + 1 \right)^{\beta-1}\cdot n^2 = O(n^{2-\epsilon})$, which verifies condition (b).

\textbf{Case 1.2}: $\delta<1$ and $\beta=1$. In this case condition ($\btilde$) simplifies to
\[
\tilde{N}^{\delta-1} n^2 = n^{2-\epsilon} \Leftrightarrow \tilde{N} = n^{\frac{\epsilon}{1-\delta}},
\]
where $1-\delta > 0$ holds thanks to $\delta < 1$. Condition ($\atilde$) and condition ($\btilde$) both concern $\tilde{N}$, and by combining them we obtain:
\[
n^{\frac{\epsilon}{1-\delta}} = n^{\frac{1-\epsilon'}{\alpha-1}} \Leftrightarrow \frac{\epsilon}{1-\delta} = \frac{1-\epsilon'}{\alpha-1} \Leftrightarrow \epsilon = \frac{1-\epsilon'}{\alpha-1}(1-\delta).
\]
We already know that $0<\frac{1-\epsilon'}{\alpha-1}\leq1$, which guarantees that $0<\epsilon\leq1-\delta$ and also verifies condition~(d). Indeed, condition (d) requires $0\leq \frac{\epsilon}{1-\delta} \leq1$, but this is already kept in check by the fact that $\frac{\epsilon}{1-\delta} = \frac{1-\epsilon'}{\alpha-1}$. Since $\delta < 1$, we have $1-\delta > 0$, and hence we can conclude that all conditions ($\atilde$), ($\btilde$) and (d) hold.

Using Fact~\ref{fact:ceil_to_bigO} we can prove that when choosing $N$ as in (c) condition (a) holds.
\begin{align*}
N^{\alpha-1}n &= \lceil n^\frac{1-\epsilon'}{\alpha-1} \rceil^{\alpha-1}n\\
&=O(n^{\frac{1-\epsilon'}{\alpha-1}{\alpha-1}}n)\\
&=O(n^{2-\epsilon'}).
\end{align*}
Observing that $\beta=1$ makes condition (b) simplify to $N^{\delta-1}n^2=O(n^{2-\epsilon})$, and we can perform a similar analysis to obtain $N^{\delta-1}n^2 \leq \left( n^\frac{\epsilon}{1-\delta}+1 \right)^{\delta-1}\cdot n^2 = O(n^{2-\epsilon})$, which verifies condition (b).

\textbf{Case 1.3}: $\delta=1$ and $\beta<1$. Similarly to the previous case, from condition ($\btilde$) we get:
\[
\tilde{M}^{\beta-1} n^2 = n^{2-\epsilon} \Leftrightarrow \tilde{M} = n^{\frac{\epsilon}{1-\beta}}.
\]
Here, condition (d) is equivalent to $0\leq\frac{\epsilon}{1-\beta}\leq1$, which is guaranteed by choosing $\epsilon$ such that $0<\epsilon\leq1-\beta$. Thus, conditions ($\atilde$), ($\btilde$) and (d) hold. Assuming condition (c) we can perform a similar analysis to the previous case and conclude that conditions (a) and (b) also hold.

\end{document}